\documentclass[11pt, letterpaper]{article}

\usepackage[margin=1in]{geometry}

\usepackage{mathtools, amssymb, amsthm, fontawesome5}
    \def\NN{\mathbf N}
    
    \def\PP{\mathbb P}
    \def\TT{\mathcal T}
    \def\geq{\geqslant}
    \def\leq{\leqslant}
    \newtheorem{theorem}{Theorem}
    \newtheorem{lemma}[theorem]{Lemma}
    \newtheorem{remark}[theorem]{Remark}
    
    \newtheorem{conjecture}[theorem]{Conjecture}
    \newtheorem{proposition}[theorem]{Proposition}

\usepackage{tikz, booktabs}
    \usetikzlibrary{arrows.meta, fadings, shapes.geometric, shapes.symbols}
    \tikzset{every picture/.style={line cap=round,line join=round}}

\usepackage[colorlinks, allcolors=blue!50!black]{hyperref}

\begin{document}

                                   \title
                            {MichelangeRoll: \\ 
          Sculpting Rational Distributions Exactly and Efficiently}

                                   \author
            {Jui-Hsiang Shao \\\small National Taiwan University
                                    \and
              Hsin-Po Wang \\\small National Taiwan University}
                                      
                         \def\day#1\year{\the\year}
                                 \maketitle

\begin{abstract}
    Simulating an arbitrary discrete distribution $D \in [0, 1]^n$ using fair
    coin tosses incurs trade-offs between entropy complexity and space and
    time complexity.  Shannon's theory suggests that $H(D)$ tosses are
    necessary and sufficient, but does not guarantee exact distribution.
    Knuth and Yao showed that a decision tree consumes fewer than $H(D) + 2$
    tosses for one exact sample.  Draper and Saad's recent work addresses
    the space and time aspect, showing that $H(D) + 2$ tosses, $O(n \log(n)
    \log(m))$ memory, and $O(H(D))$ operations are all it costs, where $m$ is
    the common denominator of the probability masses in $D$ and $n$ is the
    number of possible outcomes.
    
    In this paper, MichelangeRoll recycles leftover entropy to break the
    ``$+2$'' barrier.  With $O((n + 1/\varepsilon) \log(m/\varepsilon))$
    memory, the entropy cost of generating a ongoing sequence of $D$ is
    reduced to $H(D) + \varepsilon$ per sample.
\end{abstract}

\section{Introduction}

    When it comes to tossing fair coins to sample a discrete distribution
    $D$, there are several prices to pay.  The first is the expected number
    of tosses per $D$-sample.  In this regard, Shannon's source coding
    theorem can be used to show that $H(D)$ tosses are necessary and
    sufficient.  Algorithmically speaking, we can ``decompress'' a random
    binary string as if the string were the result of compressing a sequence
    of $D$-samples.    

    The decompression approach, however, often generates a distribution
    slightly different from $D$.  Knuth and Yao \cite{KnY76} constructed
    binary decision trees that consume $H(D) + 2$ tosses to sample $D$
    exactly, suggesting that ``$+2$'' is the price for exactness.  Given
    that, it is not hard to imagine that we can generate a $D^2$-sample,
    i.e., two iid $D$-samples, with $H(D^2) + 2$ tosses.  This reduces the
    entropy cost to $H(D) + 1$ per $D$-sample.
    Similarly, generating a batch of $b$ samples with only one ``$+2$'' will
    reduce the entropy cost to $H(D) + 2/b$ per sample.
    Since it can be easily amortized, it is stretching to call ``$+2$'' a
    fundamental price.  This begs the question: \emph{Is Shannon's bound more
    fundamental?}

    The problem with generating $D^b$-samples with a high $b$ to amortize
    ``$+2$'' is that we pay space and time to store and process complicated
    distributions.  More quantitatively, $n$ the number of possible outcomes
    will blow up to $n^b$ and $m$ the denominator of the probability masses
    will blow up to $m^b$.  That is to say, ``$+2$'' is more like a
    computational barrier than an exactness barrier.  On top of that, Knuth
    and Yao's tree is already not so trivial to construct even for $b = 1$.
    Hence, both $H(D)$ and $H(D) + 2$ stand as fundamental limits, just for
    different reasons.

    Subsequently, a series of papers clarify the computational complexity
    aspect of generating $D$-samples.  The most recent one is by Draper and
    Saad \cite{DrS25}; they make the decision tree ``more periodic'' so it
    fits in $O(n \log(n) \log(m))$ memory.  The new tree also limits to $H(D)
    + 2$ tosses and $O(H(D) + 1)$ operators per sample, matching the previous
    state of the art.

    This paper introduces \emph{MichelangeRoll}, which uses an
    \emph{asymmetric numeral system} to recycle leftover entropy that is
    neglected in earlier works.
    Modulo the concurrent work \cite{DrS25a},
    this is the first entropy sculptor that breaks
    the ``$+2$'' barrier without an exponential space complexity.  The
    precise parameters are stated in Table~\ref{tab:recent} as well as the
    main theorem below.

    \begin{theorem} [main]                                   \label{thm:main}
        Let $m$ and $n$ be positive integers.  Let $D$ be a distribution over
        $n$ outcomes.  Let $m$ be the common denominator of the probability
        masses in $D$.  For any small constant $\varepsilon > 0$, there
        exists an algorithm that generates an ongoing sequence of $D$-samples
        using $H(D) + \varepsilon$ fair coin tosses per samples, $O((n +
        1/\varepsilon) \log(m/\varepsilon))$ memory, and
        $O(\log(m/\varepsilon)^2 / \varepsilon)$ operations per sample.
    \end{theorem}

    While our time complexity is higher than previous works, it is only
    quadratic in the length of the description of $D$ and quasi-linear in
    $1/$the gap to entropy bound.  This is a cheap price to pay to avoid $n$
    and $m$ blowing up to $n^{O(1/\varepsilon)}$ and $m^{O(1/\varepsilon)}$,
    respectively.

    Finally, readers are referred to the concurrent work \cite{DrS25a} by
    Draper and Saad, which also breaks the ``$+2$'' barrier.  While their
    work generates a sequence of samples like ours do, the underlying
    distributions need not be the same and can depend on earlier samples.
    Their technique is similar to, but more dexterous than, ours.  See the
    remark below Conjecture~\ref{conj:-2} for more details.

    This paper is organized as follows.  Section~\ref{sec:background} reviews
    Knuth--Yao and other works.  Section~\ref{sec:asymmetric} then introduces
    asymmetric numeral system to recycle uniform distributions.  Finally,
    Section~\ref{sec:proof} proves Theorem~\ref{thm:main}.

\begin{table}
    \centering
    \caption{
        Recent works on simulating discrete distributions.  Costs are
        measured ``per sample''.
    }                                                      \label{tab:recent}
    \bigskip
    \begin{tabular}{cccc}
        \toprule
        Reference & Entropy & Space & Time \\
        \midrule
        Knuth--Yao \cite{KnY76}
        & $H(D) + 2$ & $O(m n \log n)$ & $O(H(D) + 1)$ \\
        Han--Hoshi \cite{HaH97}
        & $H(D) + 3$ & $O(m n \log n)$ & $O(H(D) + 1)$ \\
        Han--Hoshi \cite{HaH97}
        & $H(D) + 3$ & $O(n \log m)$ & $O((H(D) + 1) \log n)$ \\
        rejection-based \cite{SFR20}
        & $H(D) + 6$ & $O(n \log(n) \log(m))$ & $O(H(D) + 1)$ \\
        Draper--Saad \cite{DrS25}
        & $H(D) + 2$ & $O(n \log(n) \log(m))$ & $O(H(D) + 1)$ \\
        MichelangeRoll
        & $H(D) + \varepsilon$
        & $O((n + 1/\varepsilon) \log(m/\varepsilon))$
        & $O(\log(m/\varepsilon)^2 / \varepsilon)$ \\
        \midrule
        Concurrent work \cite{DrS25a}
        & $H(D) + \varepsilon$
        & $O((n + \log(m/\varepsilon)) \log(m))$
        & $O((n + \log(m/\varepsilon)) \log(m))$ \\
        \bottomrule
    \end{tabular}
\end{table}

\section{Knuth--Yao and Variants}                      \label{sec:background}

    From now on, $D = (p_1, p_2, \dotsc, p_n) \in [0, 1]^n$ is a distribution
    with the $p$'s being probability masses and $n$ being the number of
    possible outcomes.  Let the masses be rational numbers with a common
    denominator $m$, i.e., $mp_i$ is an integer for every $i \in \{1, 2,
    \dotsc, n\}$.

\subsection{Knuth--Yao}

    The well-known Knuth--Yao construction \cite{KnY76} samples discrete
    probability distributions using binary decision trees.  They use internal
    nodes to represent coin tosses, and each leaf corresponds to a possible
    outcome of $D$.  The tree is constructed in such a way that level $\ell$
    has $\lfloor 2^\ell p_i \rfloor \bmod 2 \in \{0, 1\}$ leaves that map to
    the $i$th outcome.  Note that, $\lfloor 2^\ell p_i \rfloor \bmod 2$ is
    the $\ell$th digit of the binary expansion of $p_i$.

    Such a tree is optimal in the following sense: Any algorithm that
    consumes random bits can be translated into a decision tree, infinite or
    not, and $p_i$ must coincide with the sum of $2^{-\text{level}}$ over all
    leaves that map to the $i$th outcome.  A decision tree that terminates as
    early as possible is the one that does not have two leaves at level $\ell
    + 1$ when it can have one leaf at level $\ell$.  So mirroring the binary
    expansion of $p_i$ is the one and only way to minimize entropy cost.  See
    Figure~\ref{fig:tree} for an example.

    Such a tree, however, wastes entropy and deviates away form Shannon's
    prediction of $H(D)$ bits per sample.  This is because, every time a leaf
    is mapped to an outcome symbol $X_t$, the level $L_t$ the leaf is at is
    forgotten.  If, instead, we remember the full history $(X_1, L_1)\,,\,
    (X_2, L_2)\,,\, \dotsc$ we can identify the leaves that were reached in
    the decision tree.  With that information we can recover the full history
    of coin tosses.  We therefore have $H(\text{coin tosses}) = H(\text{$X$'s
    and $L$'s}) \leq H(\text{$X$'s}) + H(\text{$L$'s})$.  In this viewpoint,
    ``$+2$'' is an overestimate of $H(L_t)$, as the following proposition
    implies.

    \begin{proposition}                                        \label{pro:+2}
        Let $\Lambda \subset \NN$ be a (possibly infinite) subset of natural
        numbers.  If a random variable $L \in \Lambda$ is such that $\PP\{L =
        \ell\} \propto 2^{-\ell}$ for all $\ell \in \Lambda$, then $H(L) \leq
        2$.  The equality holds iff $L$ follows the geometric distribution
        with success rate $1/2$ (which won't happen on a Knuth--Yao tree).
    \end{proposition}

    \begin{proof}
        Let $\lambda$ be the least element of $\Lambda$.  Let $B$ be the
        indicator of $L > \lambda$.  Then $B$ follows a Bernoulli
        distribution with mean $< 1/2$.  Conditioning on $B = 1$, we observe
        that $L - \lambda$ is again a random variable whose pmf is
        proportional to $2^{-\ell}$ for all $\ell$ in the support.  So we
        obtain a recursive upper bound on $H(L)$:
        \begin{equation}
            H(L) \leq H(B) + \PP\{B = 1\} \cdot \eqref{recur}   \label{recur}
        \end{equation}
        With $H(B) \leq 1$ and $\PP\{B = 1\} \leq 1/2$, the fixed point is
        found to be $2$, proving $H(L) \leq 2$.
    \end{proof}

    Another problem with the tree is its high memory footprint even when the
    denominator is fairly small.  Consider an example $D \coloneqq (1/947,
    946/947)$.  It takes ten bits to store $947$.  But the binary expansion
    of $1/947$ does not repeat itself in the first $946$ digits.  This
    implies that any optimal decision tree cannot repeat itself before level
    $946$, and so everything in between needs to be stored (or efficiently
    computed on demand).

    In fact, a straightforward counting argument \cite{ErM99} shows that, for
    almost all primes $m$, the binary expansion of $1/m$ does not repeat
    itself in the first $\sqrt{m} / \ln m$ bits.  Therefore, there is almost
    always an exponential gap between storing $m$ versus storing $1/m$.

    On the bright side, at least we know that the tree will eventually repeat
    itself in $m$ or fewer levels.  We can prune the repeating part of the
    tree and place a ``goto'' arrow that points to the root of the repeated
    part, as shown in Figure~\ref{fig:goto}.  Through this, any tree
    constructed out of a rational distribution can fit into a finite amount
    of memory.  While describing $n$ probability masses requires $O(n \log
    m)$ memory, the tree's memory might grow like $mn \log n$.  (Each level
    has $O(n)$  nodes and leaves; each edge needs $O(\log n)$ bits.)

\begin{figure}
    \centering
    \begin{tikzpicture}
        \def~#1{\hbox to 0.5cm{\hfil#1\hfil}}
        \draw [green!30!gray] (0, 1.5) node [right]
            {$\log_8(\sec20^\circ)=$~0~.~0~0~0~0~0~1~1~1~1~0~1~0~1~0~0~0~0};
        \draw [green!30!gray] (0, 1) node [right]
            {$\log_8(\sec40^\circ)=$~0~.~0~0~1~0~0~0~0~0~1~1~0~0~1~1~1~1~1};
        \draw [green!30!gray] (0, 0.5) node [right]
            {$\log_8(\sec80^\circ)=$~0~.~1~1~0~1~0~1~1~1~1~0~0~0~1~0~0~0~0};
        \tikzset{shift={(3.3, -0.1)}}
        \filldraw [brown]
            (0, 0) circle (1pt)
            (0, 0) -- (0.5, 0.5)
            (0, 0) -- (0.5, 0)
            (0.5, 0) circle (1pt)
            (0.5, 0) -- (1, 0.5)
            (0.5, 0) -- (1, 0)
            (1, 0) circle (1pt)
            (1, 0) -- (1.5, 1)
            (1, 0) -- (1.5, 0)
            (1.5, 0) circle (1pt)
            (1.5, 0) -- (2, 0.5)
            (1.5, 0) -- (2, 0)
            (2, 0) circle (1pt)
            (2, 0) -- (2.5, 0)
            (2, 0) -- (2.5, -0.5)
            (2.5, 0) circle (1pt)
            (2.5, -0.5) circle (1pt)
            (2.5, 0) -- (3, 1.5)
            (2.5, 0) -- (3, 0.5)
            (2.5, -0.5) -- (3, 0)
            (2.5, -0.5) -- (3, -0.5)
            (3, 0) circle (1pt)
            (3, -0.5) circle (1pt)
            (3, 0) -- (3.5, 1.5)
            (3, 0) -- (3.5, 0.5)
            (3, -0.5) -- (3.5, 0)
            (3, -0.5) -- (3.5, -0.5)
            (3.5, 0) circle (1pt)
            (3.5, -0.5) circle (1pt)
            (3.5, 0) -- (4, 1.5)
            (3.5, 0) -- (4, 0.5)
            (3.5, -0.5) -- (4, 0)
            (3.5, -0.5) -- (4, -0.5)
            (4, 0) circle (1pt)
            (4, -0.5) circle (1pt)
            (4, 0) -- (4.5, 1.5)
            (4, 0) -- (4.5, 1)
            (4, -0.5) -- (4.5, 0.5)
            (4, -0.5) -- (4.5, 0)
            (4.5, 0) circle (1pt)
            (4.5, 0) -- (5, 1)
            (4.5, 0) -- (5, 0)
            (5, 0) circle (1pt)
            (5, 0) -- (5.5, 1.5)
            (5, 0) -- (5.5, 0)
            (5.5, 0) circle (1pt)
            (5.5, 0) -- (6, 0)
            (5.5, 0) -- (6, -0.5)
            (6, 0) circle (1pt)
            (6, -0.5) circle (1pt)
            (6, 0) -- (6.5, 1.5)
            (6, 0) -- (6.5, 1)
            (6, -0.5) -- (6.5, 0.5)
            (6, -0.5) -- (6.5, 0)
            (6.5, 0) circle (1pt)
            (6.5, 0) -- (7, 1)
            (6.5, 0) -- (7, 0)
            (7, 0) circle (1pt)
            (7, 0) -- (7.5, 1.)
            (7, 0) -- (7.5, 0)
            (7.5, 0) circle (1pt)
            (7.5, 0) -- (8, 1)
            (7.5, 0) -- (8, 0)
            (8, 0) circle (1pt)
            (8, 0) -- (8.5, 1)
            (8, 0) -- (8.5, 0)
            (8.5, 0) circle (1pt)
            (8.5, 0) -- (8.9, 0)
        ;
        \fill [white, path fading=west] (6.5, 2) rectangle (9, -0.5);
        \fill [white, path fading=west] (7, 2) rectangle (9, -0.5);
        \fill [white, path fading=west] (7.5, 2) rectangle (9, -0.5);
        \fill [white, path fading=west] (8, 2) rectangle (9, -0.5);
        \fill [white, path fading=west] (8.5, 2) rectangle (9, -0.5);
    \end{tikzpicture}\nobreak\hskip -1 cm $ $
    \caption{
        Knuth--Yao's optimal tree in three steps: Step one: Take numbers that
        sum to $1$.  Step two: Compute their binary expansions.  Step three:
        Use the $1$'s in the binary expansions as leaves.
    }                                                        \label{fig:tree}
\end{figure}

\subsection{Variants for uniform distributions}

    Many works have since then attempted to simplify the tree generation
    process or work with implicit trees.  For instance, Lumbroso \cite{Lum13}
    (especially Theorem~3 therein) and Huber--Vargas \cite{HuV24} considered
    the useful special case where $D$ is uniform.  This reduces space and
    time complexity as all $p$'s are now the same.
    
    More concretely, if $D$ is the uniform distribution on $\{1, 2, \dotsc,
    n\}$, then $D^b$ is a uniform distribution on $\{1, 2, \dotsc, n^b\}$.
    From here one can simply apply Knuth--Yao, which boils down to expressing
    $1/n^b$ in binary.  Now compare this to our claim in the introduction
    that the space complexity grows exponentially in $b$: There are indeed
    $n^b$ possible outcomes and $n^b$ probability masses to be stored; but
    they are all the same, so the space complexity is only $O(\log(n^b)) =
    O(b \log n)$.

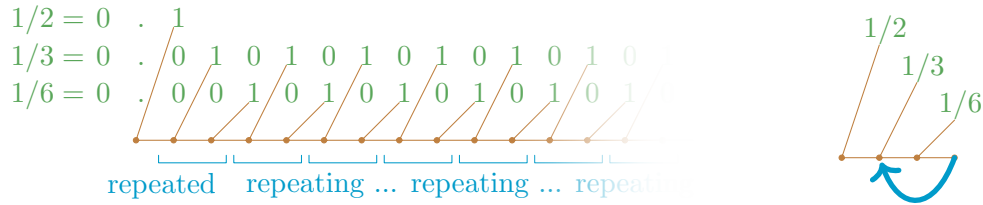
\begin{figure}
    \centering
    \begin{tikzpicture}
        \def~#1{\hbox to 0.5cm{\hfil#1\hfil}}
        \draw [green!30!gray] (0, 1.5)
            node [right] {$1/2 =$~0~.~1};
        \draw [green!30!gray] (0, 1)
            node [right] {$1/3 =$~0~.~0~1~0~1~0~1~0~1~0~1~0~1~0~1};
        \draw [green!30!gray] (0, 0.5)
            node [right]{$1/6 =$~0~.~0~0~1~0~1~0~1~0~1~0~1~0~1~0};
        \tikzset{shift={(1.8, -0.1)}}
        \filldraw [brown]
            (0, 0) circle (1pt)
            (0, 0) -- (0.5, 1.5)
            (0, 0) -- (0.5, 0)
            (0.5, 0) circle (1pt)
            (0.5, 0) -- (1, 1)
            (0.5, 0) -- (1, 0)
            (1, 0) circle (1pt)
            (1, 0) -- (1.5, 0.5)
            (1, 0) -- (1.5, 0)
            (1.5, 0) circle (1pt)
            (1.5, 0) -- (2, 1)
            (1.5, 0) -- (2, 0)
            (2, 0) circle (1pt)
            (2, 0) -- (2.5, 0.5)
            (2, 0) -- (2.5, 0)
            (2.5, 0) circle (1pt)
            (2.5, 0) -- (3, 1)
            (2.5, 0) -- (3, 0)
            (3, 0) circle (1pt)
            (3, 0) -- (3.5, 0.5)
            (3, 0) -- (3.5, 0)
            (3.5, 0) circle (1pt)
            (3.5, 0) -- (4, 1)
            (3.5, 0) -- (4, 0)
            (4, 0) circle (1pt)
            (4, 0) -- (4.5, 0.5)
            (4, 0) -- (4.5, 0)
            (4.5, 0) circle (1pt)
            (4.5, 0) -- (5, 1)
            (4.5, 0) -- (5, 0)
            (5, 0) circle (1pt)
            (5, 0) -- (5.5, 0.5)
            (5, 0) -- (5.5, 0)
            (5.5, 0) circle (1pt)
            (5.5, 0) -- (6, 1)
            (5.5, 0) -- (6, 0)
            (6, 0) circle (1pt)
            (6, 0) -- (6.5, 0.5)
            (6, 0) -- (6.5, 0)
            (6.5, 0) circle (1pt)
            (6.5, 0) -- (7, 1)
            (6.5, 0) -- (7, 0)
            (7, 0) circle (1pt)
            (7, 0) -- (7.4, 0)
        ;
        \draw [cyan!80!black]
            (0.3, -0.2) -- (0.3, -0.3) -- (1.2, -0.3)
            node [below left] {repeated}
            node [below right] {\ repeating ... repeating ... repeating}
            -- (1.2, -0.2)
            (1.3, -0.2) -- (1.3, -0.3) -- (2.2, -0.3) -- (2.2, -0.2)
            (2.3, -0.2) -- (2.3, -0.3) -- (3.2, -0.3) -- (3.2, -0.2)
            (3.3, -0.2) -- (3.3, -0.3) -- (4.2, -0.3) -- (4.2, -0.2)
            (4.3, -0.2) -- (4.3, -0.3) -- (5.2, -0.3) -- (5.2, -0.2)
            (5.3, -0.2) -- (5.3, -0.3) -- (6.2, -0.3) -- (6.2, -0.2)
            (6.3, -0.2) -- (6.3, -0.3) -- (7.2, -0.3) -- (7.2, -0.2)

        ;
        \fill [white, path fading=west] (5, 1.5) rectangle (7.5, -1);
        \fill [white, path fading=west] (5.5, 1.5) rectangle (7.5, -1);
        \fill [white, path fading=west] (6, 1.5) rectangle (7.5, -1);
        \fill [white, path fading=west] (6.5, 1.5) rectangle (7.5, -1);
        \fill [white, path fading=west] (7, 1.5) rectangle (7.5, -1);
    \end{tikzpicture}
    \hfil
    \begin{tikzpicture}
        \filldraw [brown]
            (0, 0) circle (1pt)
            (0, 0) -- (0.5, 1.5) + (0.08, 0.20) node [green!30!gray] {1/2}
            (0, 0) -- (0.5, 0)
            (0.5, 0) circle (1pt)
            (0.5, 0) -- (1, 1) + (0.08, 0.20) node [green!30!gray] {1/3}
            (0.5, 0) -- (1, 0)
            (1, 0) circle (1pt)
            (1, 0) -- (1.5, 0.5) + (0.08, 0.20) node [green!30!gray] {1/6}
            (1, 0) -- (1.5, 0)
            (1.5, 0) circle (1pt)
        ;
        \draw [cyan!80!black, line width=2pt, ->, shorten >=2pt]
            (1.5, 0) to [bend left=70, looseness=2] (0.5, 0)
        ;
    \end{tikzpicture}
    \caption{
        Left: A tree repeats itself.  Right: The repeating part is simplified
        by a ``goto'' arrow.
    }                                                        \label{fig:goto}
\end{figure}

\subsection{Variants for nonuniform distributions}

    For generic distributions, Han and Hoshi \cite{HaH97} used the inverse
    function of cdf to implement an easier-to-understand decision tree.
    First, the unit interval $[0, 1]$ is partitioned into $n$ subintervals,
    each corresponding to an $i \in \{1, \dotsc, n\}$ and having length
    $p_i$.  The algorithm then generates a random number in $[0, 1]$ by
    revealing its binary expansion bit by bit.  So, at any finite time, the
    random number is \emph{fuzzy}; it can be understood as an interval $[0.r,
    0.r\overline1]$, where $r$ is the bits revealed so far.  Han and Hoshi's
    algorithm will accept $[0.r, 0.r\overline1]$ if it falls completely
    within the $i$th subinterval for some $i \in \{1, \dotsc, n\}$.  When
    that happens, the next sample is $i$.  See Figure~\ref{fig:interval} for
    an illustration.

    From the description one can see that, at each level, there could be
    zero, one, or two fuzzy numbers that fall within the $i$th subinterval,
    but never three.  This is because any three consecutive fuzzy numbers
    must contain two that merge into a fuzzier number that should have been
    accepted in the previous level.  Note that this is slightly worse than
    Knuth--Yao, wherein each level has at most one leaf that maps to the
    $i$th outcome.  With computations like the following proposition, Han and
    Hoshi's algorithm is shown to use $H(D) + 3$ tosses per sample.

    \begin{proposition}                                        \label{pro:+3}
        Let $\Gamma \subset \NN \times \{\spadesuit, \clubsuit\}$ be a
        (possibly infinite) subset.  If random variables $(L, S) \in \Gamma$
        are such that $\PP\{(L, S) = (\ell, s)\} \propto 2^{-\ell}$ for all
        $(\ell, s) \in \Gamma$, then $H(L, S) \leq 3$.
    \end{proposition}

    \begin{proof}
        Generalize the idea of Proposition~\ref{pro:+2}: Let $\lambda$ be
        the least possible $\ell$ among all $(\ell, s) \in \Gamma$.  Then
        \begin{equation}
            H(L, S) \leq \frac32 + \frac12 \cdot
            \max\{\eqref{recur3}, \eqref{recur4}\}             \label{recur3}
        \end{equation}
        if both
        $(\lambda, \spadesuit)$ and $(\lambda, \clubsuit)$ are in $\Gamma$.
        if only one of them is in $\Gamma$, then
        \begin{equation}
            H(L, S) \leq 1 + \frac23 \cdot
            \max\{\eqref{recur3}, \eqref{recur4}\}             \label{recur4}
        \end{equation}
        The fixed point is $3$.
    \end{proof}

\begin{figure}
    \centering
    \begin{tikzpicture}
        \fill [cyan!20!white] (0, 12/3) rectangle (9.5, 8/3);
        \fill [yellow!50!white] (0, 8/3) rectangle (9.5, 4/3);
        \fill [magenta!20!white] (0, 4/3) rectangle (9.5, 0);
        \draw (9.5, 12/3) node [right] {$1$};
        \draw (9.5, 8/3) node [right] {$2/3$};
        \draw (9.5, 4/3) node [right] {$1/3$};
        \draw (9.5, 0) node [right] {$0$};
        \draw
            (1, 4) rectangle (1.5, 0)
            (2, 4) rectangle (2.5, 2)
            (2, 2) rectangle (2.5, 0)
            (3, 4) rectangle (3.5, 3)
            (3, 3) rectangle (3.5, 2)
            (3, 2) rectangle (3.5, 1)
            (3, 1) rectangle (3.5, 0)
            (4, 3) rectangle (4.5, 2.5)
            (4, 2.5) rectangle (4.5, 2)
            (4, 2) rectangle (4.5, 1.5)
            (4, 1.5) rectangle (4.5, 1)
            (5, 3) rectangle (5.5, 2.75)
            (5, 2.75) rectangle (5.5, 2.5)
            (5, 1.5) rectangle (5.5, 1.25)
            (5, 1.25) rectangle (5.5, 1)
            (6.5, 2.75) rectangle (6, 2.625)
            (6.5, 2.625) rectangle (6, 2.5)
            (6.5, 1.5) rectangle (6, 1.375)
            (6.5, 1.375) rectangle (6, 1.25)
            (7.5, 2.75) rectangle (7, 2.6875)
            (7.5, 2.6875) rectangle (7, 2.625)
            (7.5, 1.375) rectangle (7, 1.3125)
            (7.5, 1.3125) rectangle (7, 1.25)
            (8.5, 2.6875) rectangle (8, 2.65625)
            (8.5, 2.65625) rectangle (8, 2.625)
            (8.5, 1.375) rectangle (8, 1.34375)
            (8.5, 1.34375) rectangle (8, 1.3125)
        ;
    \end{tikzpicture}
    \caption{
        Han and Hoshi's construction with $D \coloneqq (1/3, 1/3, 1/3)$ as a
        running example.  The unit interval $[0, 1]$ is partitioned into $3$
        subintervals of length $1/3$, represented by the three colored
        strips.  The process of revealing the fuzzy random number is
        represented by rectangles with progressively halving heights.  A
        rectangle is accepted if it is monochromatic.  Accepted rectangles
        will not be divided further.
    }                                                    \label{fig:interval}
\end{figure}
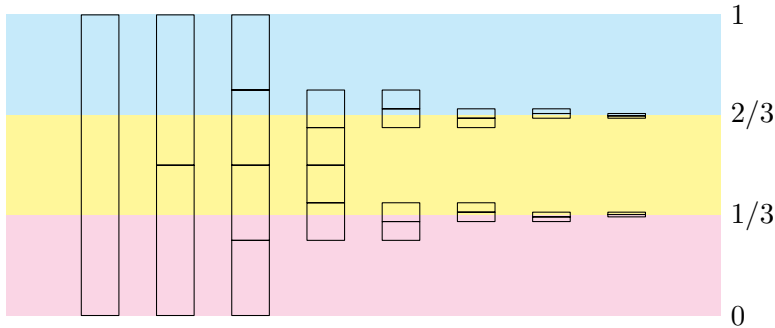

\subsection{Acceptance--rejection approach}

    Earlier this year, Draper and Saad \cite{DrS25} proposed a
    low-complexity tree based on an earlier work by Saad, Freer, Rinard, and
    Mansinghka \cite{SFR20}.  We first introduce the base work.

    Saad, Freer, Rinard, and Mansinghka \cite{SFR20} borrowed ideas from the
    acceptance--rejection framework.  They consider the least power $2^k$
    that is $\geq m$, and generate a uniform sample $U \in \{0, \dotsc, 2^k -
    1\}$.  They then check if $U$ falls in any of the intervals
        \begin{itemize}
            \item $\bigl[0, mp_1\bigr)$,
            \item $\bigl[mp_1, m(p_1 + p_2)\bigr)$,
            \item $\bigl[m (p_1 + p_2), m (p_1 + p_2 + p_3)\bigr)$,
                    \\[-25pt] \hbox{} \nopagebreak
            \item ~ \\[-25pt] \hbox{} \nopagebreak
            \item ~ \\[-25pt] \hbox{} \nopagebreak
            \item ~ \\[-25pt] \hbox{} \nopagebreak
            \item $\bigl[m (1 - p_n - p_{n-1}), m (1 - p_n)\bigr)$
            \item $\bigl[m (1 - p_n), m\bigr)$.
        \end{itemize}
    If so, the index of the subinterval is output as the next sample.  If
    not, we say that $U$ is \emph{rejected}.  When done economically, this is
    equivalent to constructing a tree with a ``dyadized'' distribution in
    mind $\ddot D \coloneqq (mp_1/2^k, mp_2/2^k, \dotsc, mp_n/2^k, 1 -
    m/2^k)$ and redirecting the $(n+1)$th outcome back to the root.  See
    Figure~\ref{fig:reject} for a comparison.  A tree like that is extremely
    easy to store: only $O(n \log(n) \log(m))$ memory is needed.

    This approach, however, wastes entropy due to higher density of gotos.
    For instance, if $m$ is $2^{k-1} + 1$, then the rejection probability is
    $(2^{k-1} - 1) / 2^k \approx 1/2$; the ratio between rejection and
    acceptance is about $1 : 1$.  We can write down a sequence
    ``ARRARAAR...'' where R and A mean rejection and acceptance,
    respectively.  This sequence carries about $2$ bits of information per
    occurrence of A, so already this wastes $2$ bits of entropy per sample.
    Plus, when $U$ is rejected, the level of the leaf that leads to the
    rejection is forgotten, which is another $2$ bits of entropy wasted (for
    the same reason behind Proposition~\ref{pro:+2}).  Overall, this approach
    wastes $4$ more bits than Knuth--Yao, totaling to $H(D) + 6$ tosses per
    sample.

    The follow-up work by Draper and Saad \cite{DrS25} addressed the loss of
    rejection entropy.  The authors noticed that, for worst-case scenarios
    such as $m = 2^{k-1} + 1$, multiples like $15m$ and $31m$ are slightly
    smaller than $2^{k+3}$ and $2^{k+4}$, respectively.  If we use $15m$ or
    $31m$ as the denominator of the $p$'s, the rejection rate can be
    controlled under $1/16$ or $1/32$.  When the rejection rate is low, the
    entropy of the A--R sequence is close to $0$, and the leaf that leads to
    rejection is not important because it is rarely visited.

    For general $m$, one can see that a multiple like $\lfloor 4^k / m
    \rfloor \cdot m$ lies between $4^k - 2^k$ and $4^k$ and saves a great
    amount of entropy.  As the multiple increases, Draper and Saad's tree
    converges to that of Knuth and Yao.  Amplifying the denominator thus
    provides a way to interpolate between ``$+2$'' and ``$+6$''.  With
    careful computations, their work \cite{DrS25} showed that $\lfloor 4^k /
    m \rfloor \cdot m$ limits wasted entropy to $+2$ bits per sample, and
    conclude that $O(n \log(n) \log(m))$ memory is possible.

\section{Asymmetric Numeral to Recycle Uniform}        \label{sec:asymmetric}

    Around Proposition~\ref{pro:+2} in the previous section, we saw that the
    Knuth--Yao tree wastes $2$ bits of entropy for forgetting the level $L_t$
    of the leaf that leads to the $t$th sample.  In this section, we
    demonstrate how to recycle non-dyadic uniform distributions.  We begin
    with borrowing an existing coding tool.

\subsection{Asymmetric numeral system}

    The \emph{asymmetric numeral system} (ANS) \cite{Dud14} is a coding
    scheme that compresses a sequence of random variables $S_1, S_2, \dotsc,
    S_t$ in to a single integer $A_t$ by applying a very intuitive pairing
    function recursively.  In this subsection, $\NN$ is the set of
    nonnegative integers, containing $0$.

    The scheme is as follows: Suppose that $(S_1, S_2, \dotsc, S_{t-1}) \in
    \NN^{t-1}$ are already mapped to $A_{t-1} \in \NN$.  We now want to push
    $S_t \in \{0, 1, \dotsc, n - 1\}$ into $A_{t-1}$ to form $A_t$.  So we
    look up and denote the pmf of $S_t$ by $f$.  Next we partition
    nonnegative integers
    \begin{equation}
        \NN =
        \NN[0] \cup \NN[1] \cup \dotsb \cup \NN[n - 1]      \label{partition}
    \end{equation}
    such that the density of $\NN[s]$ is about $f(s)$.  Finally, let $A_t$ be
    the $A_{t-1}$th element of $\NN[S_t]$.

    The art lies in the choice of the partition \eqref{partition}.  For
    instance, one can draft a partition $\NN[s] \leftarrow
    \{\lfloor1/f(s)\rfloor, \lfloor2/f(s)\rfloor, \lfloor3/f(s)\rfloor,
    \dotsc\}$ and resolve collisions locally.  This way, $A_t / A_{t-1}$ will
    be approximately $1 / f(S_t)$, so the number of digits $A_t$ has
    increases by $f(S_t) \log(1 / f(S_t))$.  The expectation of the increment
    is $H(f)$, so ANS is asymptotically optimal.

    MichelangeRoll only needs to recycle uniform distributions, as
    Section~\ref{sec:proof} will explain.  So its partition \eqref{partition}
    is as simple as $\NN[s] \coloneqq \{ an + s \,|\, a \in \NN \}$, and its
    encoding scheme is as simple as $A_t \coloneqq A_{t-1}n + S_t$.  What is
    good about ANS is its ability to recycle uniform distributions with
    distinct $n$'s, as discussed in the next subsection.

\begin{figure}
    \centering
    \begin{tikzpicture}
        \filldraw [brown]
            (0, 0) circle (1pt) % 3 + 2 / 5
            (0, 0) -- (0.5, 1) + (0.08, 0.20) node [green!30!gray] {3/5}
            (0, 0) -- (0.5, 0)
            (0.5, 0) circle (1pt) % 1 + 4 / 5
            (0.5, 0) -- (1, 0.5) + (0.08, 0.20) node [green!30!gray] {2/5}
            (0.5, 0) -- (1, 0)
            (1, 0) circle (1pt) % 2 + 3 / 5
            (1, 0) -- (1.5, 0.5) + (0.08, 0.20) node [green!30!gray] {2/5}
            (1, 0) -- (1.5, 0)
            (1.5, 0) circle (1pt) % 4 + 1 / 5
            (1.5, 0) -- (2, 1) + (0.08, 0.20) node [green!30!gray] {3/5}
            (1.5, 0) -- (2, 0)
            (2, 0) circle (1pt) 
        ;
        \draw [cyan!80!black, line width=2pt, ->, shorten >=2pt]
            (2, 0) to [out=-90, in=0, looseness=1] (1, -1)
            (1, -1) to [out=180, in=-90, looseness=1] (0, 0)
        ;
    \end{tikzpicture}
    \hfil
    \begin{tikzpicture}
        \filldraw [brown]
            (0, 0) circle (1pt)
            (0, 0) -- (0.5, 0)
            (0, 0) -- (0.5, -0.5)
            (0.5, 0) circle (1pt)
            (0.5, -0.5) circle (1pt)
            (0.5, 0) -- (1, 1) + (0.08, 0.20) node [green!30!gray] {3/5}
            (0.5, 0) -- (1, 0.5) + (0.08, 0.20) node [green!30!gray] {2/5}
            (0.5, -0.5) -- (1, 0)
            (0.5, -0.5) -- (1, -0.5)
            (1, 0) circle (1pt)
            (1, -0.5) circle (1pt)
            (1, 0) -- (1.5, 1) + (0.08, 0.20) node [green!30!gray] {3/5}
            (1, 0) -- (1.5, 0)
            (1.5, 0) circle (1pt)
        ;
        \draw [cyan!80!black, line width=2pt, ->, shorten >=2pt]
            (1.5, 0) to [out=-45, in=0, looseness=1.5] (1, -1)
            (1, -0.5) to [out=0, in=0, looseness=2] (1, -1)
            (1, -1) to [out=180, in=-90, looseness=1] (0, 0)
        ;
    \end{tikzpicture}
    \caption{
        The goal is to generate $D \coloneqq (3/5, 2/5)$.  Left: Follow
        Knuth--Yao's recipe and wait patiently for repetition.  Right:
        Generate $\ddot D \coloneqq (3/8, 2/8, 3/8)$ and reject the third
        outcome.
    }                                                      \label{fig:reject}
\end{figure}
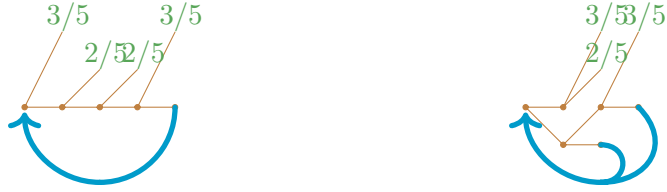

\subsection{Aggregate uniform distributions}

    Suppose that $N_1, S_1, N_2, S_2, N_3, S_3, \dotsc$ is a sequence of
    integer random variables.  $N_t$ is independent of what comes before it.
    $S_t \in \{0, 1, \dotsc, N_t - 1\}$ follows the uniform distribution
    conditioning on what comes before $N_t$.  The goal of this subsection is
    to collect entropy from the $S$'s using ANS.

    \begin{lemma}                                         \label{lem:uniform}
        Define $A_0 \coloneqq 0$ and $A_t \coloneqq A_{t-1}N_t + S_t$.
        Conditioning on $N_1, N_2, \dotsc, N_t$, the aggregated integer $A_t
        \in \{0, 1, \dotsc, N_1N_2\dotsm N_t - 1\}$ is uniform.
    \end{lemma}

    \begin{proof}
        Apply induction on $t$.  The base case is trivial.
        Suppose, for the induction step, that $A_{t-1}$ is uniform in $\{0,
        1, \dotsc, N_1N_2\dotsm N_{t-1} - 1\}$ conditioning on $N_1, N_2,
        \dotsc, N_{t-1}$.  We want two things:
        \begin{itemize}
            \item $A_{t-1}$ is uniform conditioning on $N_1, N_2, \dotsc,
                N_t$ (not just up to $N_{t-1}$).
            \item $S_t$ is uniform conditioning on $N_1, N_2, \dotsc, N_t$.
        \end{itemize}
        For the first bullet, since $N_t$ is independent of what comes before
        it, further conditioning on $N_t$ does not alter the distribution of
        $A_{t-1}$.  The second bullet follows from our assumption on $S_t$.
        Together, the bijective map $A_t \leftrightarrow N_tA_{t-1} + S_t$
        leads to a uniform sample $A_t \in \{0, 1, \dotsc, N_1N_2\dotsm N_t -
        1\}$ conditioning on $N_1, N_2, \dotsc, N_t$.  This finishes the
        proof.
    \end{proof}

\subsection{Recycle uniform into fair coin tosses}

    In this subsection, we show how to turn $A_t$ into fair and independent
    coin tosses conditioning on $N_1, N_2, \dotsc, N_t$.  As preparation,
    initialize $N$ as $N_1N_2\dotsm N_t$ and $A$ as $A_t$.  Now enter the
    main loop:
    \begin{itemize}
        \item \tikz [overlay, remember picture]
            \coordinate (start) at (-0.5, 0.3);%
            If $N$ is odd and $A = N - 1$, terminate.
        \item Output the parity of $A$.
        \item Divide both $N$ and $A$ by $2$ and discard the remainders.
        \item \tikz [overlay, remember picture]
            \draw [cyan!80!black, line width=2pt, ->]
            (-.5, -0.1) to [bend left=80, looseness=1.5] (start);%
            Go back to the first bullet.
    \end{itemize}
    This generates a sequence of coin tosses, as depicted in
    Figure~\ref{fig:ANS}.  Despite that we do not know how many tosses it
    will yield, the tosses will be independent and fair, thanks to the
    following lemma.

\begin{figure}
    \centering
    \begin{tikzpicture}
        \foreach \x in {0, 1, 2, 3, 4}{
            \draw [rounded corners=3pt]
                (2/3 + \x, 2) rectangle +(0.5, 1) +(0.25, 0.75) node {$\x$};
        }
        \foreach \Y [count=\y] in {\spadesuit, \clubsuit, \diamondsuit}{
            \draw [rounded corners=3pt]
                (7 + \y, 2) rectangle +(0.5, 1) +(0.25, 0.25) node {$\Y$};
        }
        \foreach \x in {0, 1, 2, 3, 4}{
            \foreach \Y [count=\y] in {\spadesuit, \clubsuit, \diamondsuit}{
                \draw [rounded corners=2pt]
                    (\x*2 + \y*2/3, 0) rectangle +(0.5, 1)
                    +(0.25, 0.75) node {$\x$}
                    +(0.25, 0.25) node {$\Y$}
                ;
            }
        }
        \tikzset{A/.style={line width=2pt, ->, cyan!80!black}}
        \draw [A]
            (5.5, 2.5) to [out=0, in=90, looseness=1.5] (6.5+1/12, 1.2)
            (7+2/3, 2.5) to [out=180,in=90,looseness=1.5] (6.5+1/12, 1.2)
        ;
        \draw [A] (2/3, -0.2) -- (5+2/3, -0.2) to [out=0, in=45] (5, -1);
        \draw [A] (6, -0.2) -- (8 + 1/3, -0.2) to [out=0, in=45] (7+2/3, -1);
        \draw [A] (8+2/3, -0.2) -- (9+2/3, -0.2) to [out=0, in=45] (9, -1);
        \tikzset{
            david/.pic={
                \tikzset{scale=0.15}
                \draw (1, 1) |- (0.8, -0.5) |- (0, -1);
                \draw (0.5, 0) -- (0.8, 0);
                \draw (0.5, -1) -- (0.5, -1.2);
                \draw (0.6, -0.8) -- (0.8, -0.8);
                \filldraw [yellow!80!gray, line width=0.1cm]
                (1, 1) -- (-1, 1) -- (-1, -1) -- cycle;
            },
            1bit/.pic={
                \node [brown!40!yellow] {\huge\bfseries 1};
                \node [brown!30!yellow]
                    {\small b\textcolor{yellow}{i}t\kern1.35pt};
            }
        }
        \tikzset{left color=yellow, right color=white}
        \shadedraw [brown] (2.5, -1.5) circle (1/3) pic {david};
        \shadedraw [brown] (3.5, -1.5) circle (1/3) pic {1bit};
        \shadedraw [brown] (4.5, -1.5) circle (1/3) pic {david};
        \shadedraw [brown] (6 + 1/6, -1.5) circle (1/3) pic {1bit};
        \shadedraw [brown] (7 + 1/6, -1.5) circle (1/3) pic {david};
        \shadedraw [brown] (8.5, -1.5) circle (1/3) pic {1bit};
    \end{tikzpicture}
    \caption{
        A visualization of ANS: Two uniform distributions are combined into
        one before being dyadized into coin tosses.  Unlike \cite{HuV24}, ANS
        avoids paying ``$-3$'' twice (cf.\ Lemma~\ref{lem:-3}).
    }                                                         \label{fig:ANS}
\end{figure}
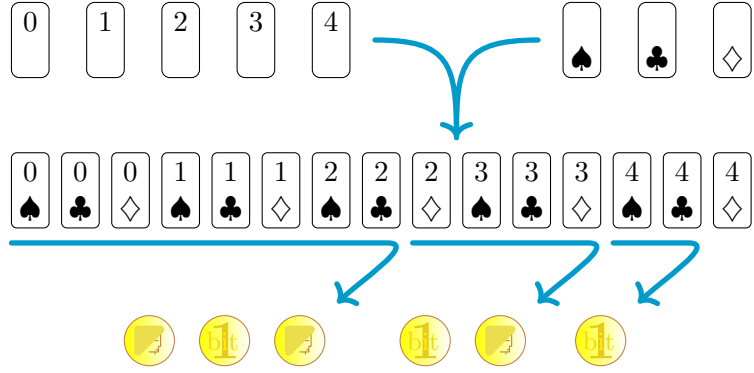

    \begin{lemma}                                            \label{lem:fair}
        Let $N$ be a positive integer, and let $A \in \{0, 1, \dotsc, N -
        1\}$ be uniform.  (A) If $N$ is even, the parity of $A$ is fair and
        $\lfloor A/2 \rfloor \in \{0, 1, \dotsc, N/2\}$ is uniform; and they
        are independent.  (B) If $N$ is odd and $A < N - 1$, the parity of
        $A$ is fair and $\lfloor A/2 \rfloor \in \{0, 1, \dotsc, (N - 1)/2\}$
        is uniform; and they are independent.
    \end{lemma}

    \begin{proof}
        (A) is straightforward.  (B) reduces to (A).
    \end{proof}
    
    Because the information about whether $A = N - 1$ or not is not turned
    into tosses, some entropy is lost in the process.  Luckily, we know how
    to control the  loss.

    \begin{lemma}                                              \label{lem:-3}
        Conditioning on $N$, the expected number of tosses generated by the
        bulleted procedure above is $> \log_2(N) - 3$.
    \end{lemma}

    \begin{proof}
        Let $\TT(N)$ be the expected number of tosses we gain as a function
        in $N$.  By the recursive nature of the procedure, there is a
        recursive relation
        \begin{equation}
            \TT(N) = 1 + \TT\Bigl(\frac N2\Bigr)                 \label{even}
        \end{equation}
        when $N$ is even.  When $N$ is odd, there is
        \begin{equation}
            \TT(N) = \frac{N - 1}{N}
            \Bigl( 1 + \TT\Bigl(\frac{N - 1}{2}\Bigr) \Bigr).     \label{odd}
        \end{equation}
        We now run mathematical induction on a slightly stronger hypothesis:
        $\TT(N) \geq (1 + 2/N) \log_2(N) - 3$.  It remains to check
        compatibility of \eqref{even} and \eqref{odd} with the hypothesis.
        For the even case, we have
        \begin{align*}
            \TT(N)
            & = 1 + T\Bigl(\frac N2\Bigr) \\
            & \geq 1 +
                \Bigl(1 + \frac4N\Bigr) \log_2\Bigl(\frac N2\Bigr) - 3 \\
            & = \Bigl(1 + \frac2N\Bigr) \log_2(N) - 3
                + \frac2N (\log_2(N) - 2) \\
            & \geq \Bigl(1 + \frac2N\Bigr) \log_2(N) - 3.
        \end{align*}
        The last inequality holds when $N \geq 4$, proving the induction
        hypothesis for the even case.  For the odd case, we have
        \begin{align}
            \TT(N)
            & = \frac{N - 1}{N}
                \Bigl( 1 + T\Bigl(\frac{N - 1}{2}\Bigr) \Bigr) \notag\\
            & \geq \frac{N - 1}{N} \Bigl( 1 + \Bigl(1 + \frac{4}{N - 1}\Bigr)
                \log_2\Bigl(\frac{N - 1}{2}\Bigr) - 3 \Bigr) \notag\\
            & = \frac{N + 3}{N} \log_2(N - 1) - \frac{3K + 1}{N} \notag\\
            & \geq \Bigl(1 + \frac2N\Bigr)\log_2(N) - 3, \label{wait}
        \end{align}
        where the last inequality is dealt in the next lemma for $N \geq 9$.
        For $N = 2, 3, 5, 7$, we check $\TT(N) > \log_2(N) - 3$ directly.
        This finishes proving $\TT(N) \geq (1 + 2/N) \log_2(N) - 3 >
        \log_2(N) - 3$ up to the correctness of \eqref{wait}.
    \end{proof}

    \begin{lemma}
        Inequality \eqref{wait} is true for $N \geq 9$.
    \end{lemma}

    \begin{proof}
        \def\?{\stackrel{\smash?}}
        Multiply both sides by $N$ and cancel $3N$ to reduce the goal to
        \[
            (N + 3) \log_2(N - 1)
            \?\geq (N + 2)\log_2(N) + 1.
        \]
        Apply mean value theorem to $\log_2$ over $[N - 1, N]$ to reduce the
        goal to
        \[
            (N + 3) \Bigl(\log_2(N) - \frac{\log_2(e)}{N - 1}\Bigr)
            \?\geq (N + 2)\log_2(N) + 1.
        \]
        Move $\log_2(N)$ to the left and everything else to the right to
        reduce the goal to
        \[
            \log_2(N)
            \?\geq 1 + \log_2(e)\frac{N + 3}{N - 1}.
        \]
        LHS is monotonically increasing and RHS is monotonically decreasing.
        They meet at around $N = 8.46$.  So \eqref{wait} holds for $N \geq
        9$.
    \end{proof}

    So far we had showed that out of a uniform sample with $N$ possible
    outcomes, $\log_2(N) - 3$ bits of randomness can be extracted.  Numerical
    evaluations show that $\TT(N) > \log_2(N) - 2$ holds and ``$-2$'' is
    asymptotically tight.  But we have difficulty guessing induction
    hypothesis.

    \begin{conjecture}                                        \label{conj:-2}
        Conditioning on $N$, the expected number of tosses generated by the
        bulleted procedure at the beginning of this subsection is $\geq
        \log_2(N) - 2$.
    \end{conjecture}

    \begin{remark}
        In a private communication, Draper (a coauthor of
        \cite{DrS25,DrS25a}) pointed out that Lemma~\ref{lem:-3} is a
        consequence of \cite[(14)]{Eli72} and that Conjecture~\ref{conj:-2}
        can be proved by an argument Knuth and Yao used to prove the ``$+2$''
        penalty.  More details are put in Appendix~\ref{app:-2}.
    \end{remark}

    \begin{remark}
        Despite sharing common elements, the concurrent work \cite{DrS25a}
        avoids depleting $A$.  Instead, Draper and Saad extract bits only
        when the penalty term $(N - 1) / N$ in \eqref{odd} is small and break
        the loop early.  On the other hand, our approach loses a considerable
        amount of entropy when $N$ eventually becomes $15$, $7$, and $3$.
    \end{remark}

\section{Proof of the Main Theorem}                         \label{sec:proof}

\begin{figure}
    \centering
    \begin{tikzpicture}
        \tikzset{
            A/.style={
                draw=cyan!80!black, text=cyan!50!black,
                line width=2pt, ->, shorten >=2pt, shorten <=2pt
            }
        }
        \draw [rounded corners=5pt]
            (0, 3) rectangle node {Subinterval lookup} (6, 2)
            (0, 1) rectangle node {ANS} (6, 0)
        ;
        \draw [A] (-3, 3) to [bend left=10]
            node [above] {fresh tosses}(0, 3);
        \draw [A] (6, 3) to [bend left=10]
            node [above] {$D$-samples} (9, 3);
        \draw [A] (6, 2) to [bend left=90, looseness=5]
            node [right] {$N$'s and $S$'s} (6, 1);
        \draw [A] (0, 1) to [bend left=90, looseness=5]
            node [left] {recycled tosses} (0, 2);
    \end{tikzpicture}
    \caption{
        An overview of the MichelangeRoll.
    }                                                    \label{fig:overview}
\end{figure}
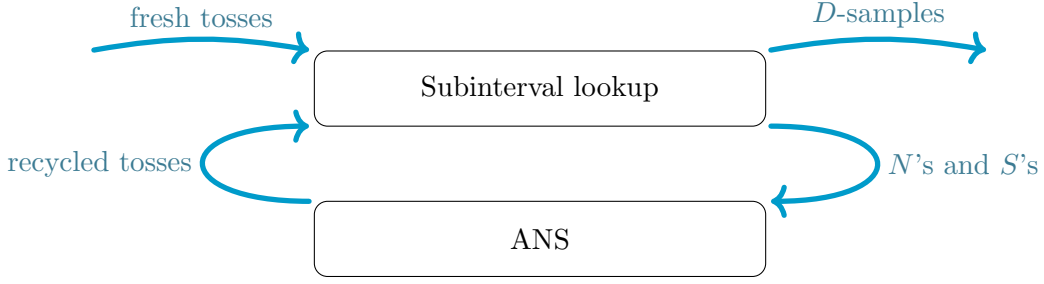

    In this section, we prove Theorem~\ref{thm:main}.  See
    Figure~\ref{fig:overview} for an overview.  Here is some initialization:
    \begin{itemize}
        \item Let $2^k$ be the least power of $2$ that is $\geq m$.
        \item Let $2^j$ be the least power of $2$ that is $\geq
            1/\varepsilon^2$.
        \item Let $M$ be $\lfloor 2^{j+k} / m \rfloor \cdot m$, a multiple of
            $m$ in the range $[(1 - \varepsilon^2) 2^{j+k}, 2^{j+k}]$
        \item Prepare a clean ANS with $A_0 = 0$ and time index $t = 1$.
    \end{itemize}
    Now enter the main loop:
    \begin{itemize}
        \item \tikz [overlay, remember picture]
            \coordinate (start) at (-0.5, 0.2);%
            At time $t$, generate a uniform integer sample $U_t$ in $[0,
            2^{j+k})$ by tossing $j + k$ coins.
        \item Among the following list of subintervals, find out the index
            $I_t$ of the one that contains $U_t$: $[0, Mp_1)$\,,\, $[Mp_1,
            M(p_1 + p_2))$\,,\, $\dotsc$\,,\, $[M(1 - p_n), M)$\,,\, $[M,
            2^{j+k})$.  (Note: the $i$th subinterval is $Mp_i$ units long.)
        \item Let $N_t$ be the length of the $I_t$th subinterval.  Let $S_t$
            be the distance between $U_t$ and the left end of the $I_t$th
            subinterval.  Recycle ($N_t, S_t$) into the ANS.
        \item If $I_t \leq n$, output $I_t$ as the next sample; if $I_t = n +
            1$, it is a rejection and nothing is output.
        \item If the product of $N$'s in the ANS becomes too big, turn the
            $A_t$ into a sequence of coin tosses that can be used to generate
            $U_t$.  Afterwards, reset the ANS by forgetting all $N$'s and
            setting $t = 0$ and $A_0 = 0$.
        \item \tikz [overlay, remember picture]
            \draw [cyan!80!black, line width=2pt, ->]
            (-.5, -0.1) to [bend left=60, looseness=1] (start);%
            Go back to the first bullet with $t$ increased by $1$.
    \end{itemize}
    It remains to check five things: One, if $I_t \leq n$, then $I_t \sim D$.
    Two, the bits ANS yields are indistinguishable from a fair coin.  Three,
    the expected consumption of fresh tosses (those not from the ANS) is
    $H(D) + \varepsilon$.  Four, the space complexity.  Five, the time
    complexity.

\subsection{MichelangeRoll generates \texorpdfstring{$D$}{D} exactly}

    The likelihood ratio between $I_t = 1$ and $I_t = 2$ is the ratio between
    the lengths of the two subintervals, which is $p_1 : p_2$.  This relation
    holds for all pairs of indices, hence $I_t$ must follow $D$ when it is
    not $n + 1$.

\subsection{MichelangeRoll recycles entropy}

    Note that $N_t$ is a function in $I_t$ and there is a bijection between
    $U_t \leftrightarrow (I_t, S_t)$.  Each time we get a new $U_t$, it has
    no memory of anything with subscript $t - 1$ and earlier.  Hence, the
    sequence $I_1, I_2, \dotsc$ are independent copies of $I_1$; the sequence
    $N_1, N_2, \dotsc$ are independent copies of $N_1$.  More importantly,
    each $S_t$ is uniform in $[0, N_t)$ conditioning on $N_t$ and is
    independent of all other $N$'s.  Now the premises of
    Lemma~\ref{lem:uniform} are satisfied, so $A_t$ is uniform in $[0,
    N_1N_2\dotsm N_t - 1)$ conditioning on $I_1, I_2, \dotsc, I_t$ and $N_1,
    N_2, \dotsc, N_t$.

    By Lemma~\ref{lem:fair}, ANS yields tosses that are fair and independent
    of each other.  Also, by Lemma~\ref{lem:uniform}, these tosses are
    independent of the $I$'s and the $N$'s, despite that the number of tosses
    may depend on the $N$'s.  Those tosses are indistinguishable from a
    fresh source of fair and independent tosses.

\subsection{MichelangeRoll breaks the ``\texorpdfstring{$+2$}{+2}'' barrier}

    Entropy is lost at two places:
    \begin{itemize}
        \item It forgets which $I_t$ are rejections.
        \item It forgets $3$ bits when turning $A$ into coin tosses, by
            Lemma~\ref{lem:-3}.
    \end{itemize}
    For the first bullet, observe that the rejection rate is $1 - M / 2^{j+k}
    \leq \varepsilon^2$.  So $H(\text{Bernoulli}(\varepsilon^2))$ bits are
    lost for each $t$.  That is, $H(\text{Bernoulli}(\varepsilon^2)) / (1 -
    \varepsilon^2)$ bits are lost for each $D$-sample.  When $\varepsilon$ is
    small, this quantity is about $O(\varepsilon^2 \log\varepsilon^{-2})$,
    which relaxes to $O(\varepsilon)$.

    For the second bullet, we lost $3$ bits every time the ANS is reset.  We
    let the ANS hold integers up to  $(j + k) / \varepsilon =
    O(\log_2(m/\varepsilon) / \varepsilon)$ bits long, so it resets once
    every $1 / \varepsilon$ samples of $S_t$.  This means that $3 \varepsilon
    / (1 - \varepsilon^2)$ bits are lost per $D$-sample.
    
    Overall, it looses $O(\varepsilon)$ bits per sample, making the total
    entropy cost $H(D) + O(\varepsilon)$.

\subsection{The space complexity}

    We want to store the prefix sums of $n$ integers $Mp_1, Mp_2, \dotsc,
    Mp_n$ so we can binary-search for $I_t$.  This costs $n (j + k) = O(n
    \log(m/\varepsilon))$ memory.  We also want to maintain $A_t$ and
    $N_1N_2\dotsm N_t$ (just the product, not individual $N$'s) of the ANS,
    as well as a buffer for the $(j + k) / \varepsilon$ extracted coin
    tosses.  This costs $O(\log(m/\varepsilon) / \varepsilon)$
    memory\footnote{ Note that the buffer is guaranteed to be emptied before
    the next reset.  This is because each iteration of the main loop consumes
    $j + k$ while multiplying $N$ by something less than $2^{j+k}$.  In
    particular, this is not just a high-probability bound on memory as in
    the queueing theory, but an almost always bound.}.  In total, we need
    $O((n + 1/\varepsilon) \log(m/\varepsilon))$ memory.

\subsection{The time complexity}

    Each $D$-sample is equivalent to about $1 / (1 - \varepsilon^2)$ samples
    of $I_t$.  Each $I_t$ is generated by a uniform sample $U_t$ in $[0,
    2^{j+k})$, which takes $O(j + k) = O(\log(m/\varepsilon))$ operations.
    Finding the index $I_t$ takes $O(\log(m/\varepsilon) \log(n))$ operations
    for binary-searching in $n$ prefix sums.  This relaxes to
    $O(\log(m/\varepsilon)^2)$ as $m/\varepsilon \geq m \geq n$.

    Now pushing $(N_t, S_t)$ into the ANS takes $O(\log(m/\varepsilon)^2 /
    \varepsilon)$ operations, for that is what schoolbook integer
    multiplication takes.  Extracting bits from the ANS takes
    $O(\log(m/\varepsilon) / \varepsilon)$ operations per reset.  Overall,
    the time complexity is $O(\log(m/\varepsilon)^2 / \varepsilon)$ per
    $D$-sample.

\subsection{Wrap up}

    So far, we have shown that MichelangeRoll uses $H(D) + O(\varepsilon)$
    coin tosses, $O((n + 1/\varepsilon) \log(m/\varepsilon))$ memory, and
    $O(\log(m/\varepsilon)^2 / \varepsilon)$ operations.  Rescaling
    $\varepsilon$ by a constant factor matches the desired statement of
    Theorem~\ref{thm:main}.

\section{Conclusion}

    This paper introduces MichelangeRoll, an entropy sculptor using
    asymmetric numeral systems to recycle leftover entropy.  By storing and
    processing entropy with integer variables, it achieves an entropy cost of
    $H(D) + \varepsilon$ per sample, breaking the complexity and exactness
    barrier ``$+2$'' established earlier.

    Our approach shows that exact simulation of discrete distributions is not
    too different from the approximate counterpart: Now that both can be
    performed with diminishing entropy losses and linear-ish space and time,
    exact simulation becomes more competitive for having cleaner theoretical
    guarantees.  This makes MichelangeRoll practical for cases like
    Bernoulli$(1/100)$, where $H(D)$ is small and so ``$+2$'' used to be a
    large overhead.

    Future work might include simplifying the implementation of ANS further,
    as well as recycling the rejection entropy from the A--R sequence.  We
    are also interested in generalizing MichelangeRoll to take non-dyadic
    uniform or biased Bernoulli as entropy source.

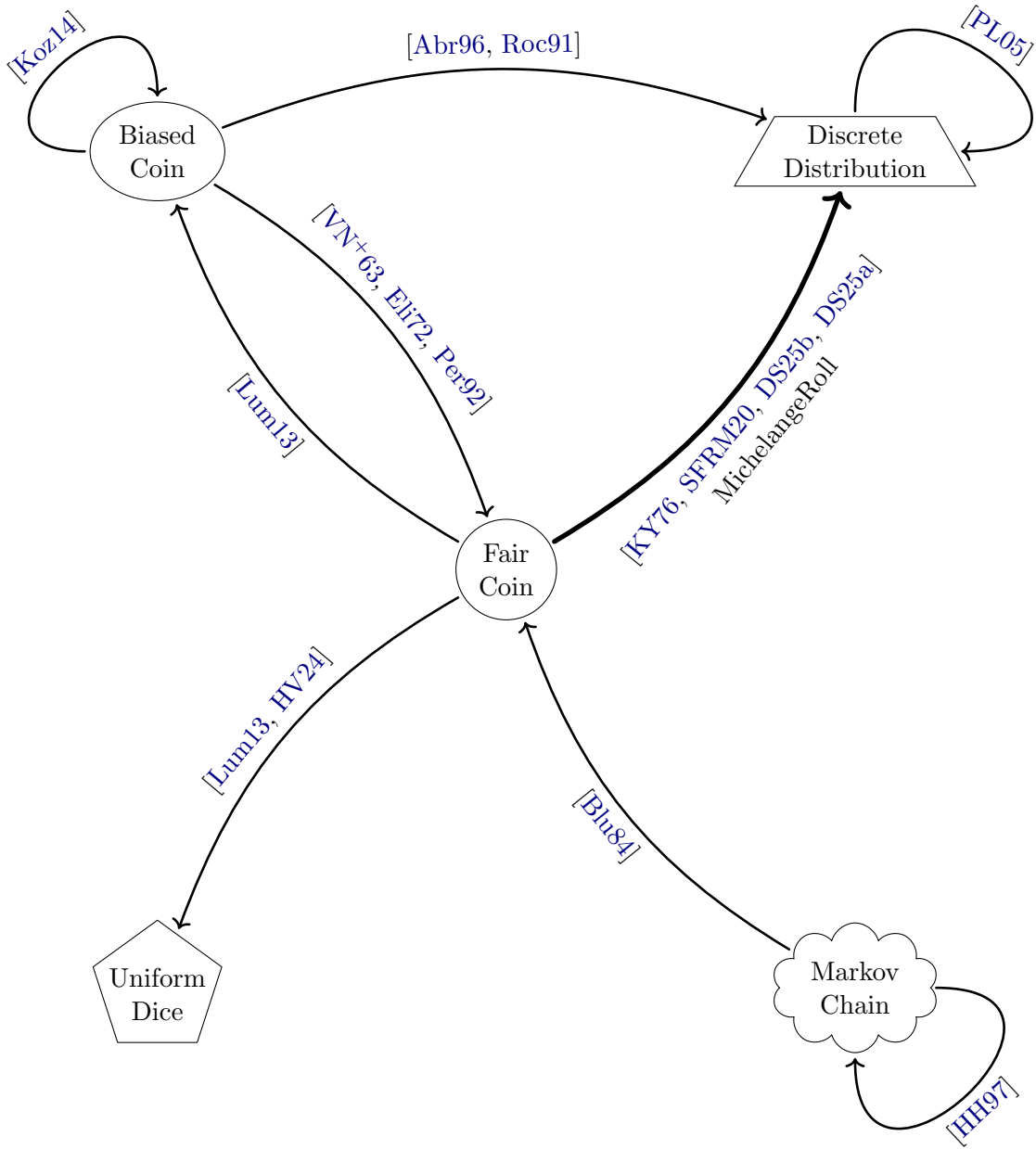
\begin{figure}
    \centering
    \tikzset{
        V/.style={draw, align=center},
        A/.style={shorten <=2pt, shorten >=2pt, line width=1pt, ->},
        N/.style={auto, sloped},
    }
    \noindent
    \kern-10cm
    \begin{tikzpicture}
        \draw
            (0, 0) node (FC) [V, circle] {Fair\\Coin}
            (-5, 6) node (BC) [V, ellipse] {Biased\\Coin}
            (5, 6) node (DD) [V, trapezium]
                {Discrete\\Distribution}
            (-5, -6) node (UD) [V, regular polygon, inner sep=-4pt]
                {\\[-1pt] Uniform\\Dice}
            (5, -6) node (MC) [V, cloud, cloud ignores aspect]
                {Markov\\Chain}
        ;
        \draw [A] (FC) to [bend left=20, '] node [N] {
                \cite{Lum13}
            } (BC);
        \draw [A] (FC) to [bend right=20] node [N] {
                \cite{Lum13, HuV24}
            } (UD);
        \draw [A, line width=2pt] (FC) to [bend right=20]
            node [N, ', align=center] {
                \cite{KnY76, SFR20, DrS25, DrS25a} \\ MichelangeRoll
            } (DD);
        \draw [A] (BC) to [bend left=20] node [N] {
                \cite{Von63, Eli72, Per92}
            } (FC);
        \draw [A] (BC) to [loop, out=180, in=90, looseness=5] node [N] {
                \cite{Koz14}
            } (BC);
        \draw [A] (BC) to [bend left=20] node [N] {
                \cite{Abr96, Roc91}
            } (DD);
        \draw [A] (DD) to [loop, out=90, in=0, looseness=5] node [N] {
                \cite{PaL05}
            } (DD);
        \draw [A] (MC) to [bend left=20] node [N, '] {
                \cite{Blu84}
            } (FC);
        \draw [A] (MC) to [loop, out=0, in=-90, looseness=5] node [N, '] {
                \cite{HaH97}
            } (MC);
            
    \end{tikzpicture}
    \kern-10cm
    \caption{
        Our classification of earlier works based on the type of entropy
        source and the type of target distribution.
    }                                                     \label{fig:digraph}
\end{figure}

\appendix

\section{Confirmation of Conjecture~\ref{conj:-2}}             \label{app:-2}

    In this appendix, we briefly go over Draper's idea (sent over a private
    communication) that confirms Conjecture~\ref{conj:-2}.

    To begin, let us recall that we extract bits by
    looking at the parity of $A$ before halving $A$ and $N$,
    and terminate when $N$ is odd and $A = N - 1$.
    This procedure can be paraphrased as follows:
    Let $k_1 > k_2 > \dotsb > k_w$ be the positions of $1$'s in the binary
    representation of $N$, i.e.,
    $N = 2^{k_1} + 2^{k_2} + \dotsb + 2^{k_w}$.
    Then we can extract $k_1$ bits if $A \in [0, 2^{k_1})$,
    $k_2$ bits if $A - 2^{k_1} \in [0, 2^{k_2})$,
    $k_3$ bits if $A - 2^{k_1} - 2^{k_2} \in [0, 2^{k_3})$,
    and so on.
    In sum, the expected number of bits extracted is
    \[
        \sum_{i=1}^{w} k_i \cdot
        \PP \{A \in \text{an interval of length } 2^{k_i}\}
        = \frac1N \sum_{i=1}^{w} k_i \cdot 2^{k_i}.
    \]
    Now the sum next to $1/N$ was studied by Knuth and Yao
    \cite[(2.19)]{KnY76} and can be expressed as
    \[ -\nu\Bigl( \sum_{i=1}^{w} 2^{k_i} \Bigr) = -\nu(N). \]
    Now \cite[Theorem~2.2]{KnY76} shows that $-N \log_2(N) \le \nu(N) \le -N
    \log_2(N) + 2N$.
    And hence the number of bits extracted is $-\nu(N)/N \ge \log_2(N) - 2$.

\section*{Acknowledgments}

    The authors thank Feras Saad for sharing unpublished manuscripts.

\bibliographystyle{alpha}
\bibliography{RecycleEntropy-2}

\newcommand{\etalchar}[1]{$^{#1}$}
\begin{thebibliography}{SFRM20}

\bibitem[Abr96]{Abr96}
J.~Abrahams.
\newblock Generation of discrete distributions from biased coins.
\newblock {\em IEEE Transactions on Information Theory}, 42(5):1541--1546,
  September 1996.

\bibitem[Blu84]{Blu84}
M.~Blum.
\newblock Independent unbiased coin flips from a correlated biased source: A
  finite state markov chain.
\newblock In {\em 25th Annual Symposium {{onFoundations}} of Computer Science,
  1984.}, pages 425--433, October 1984.

\bibitem[DS25a]{DrS25a}
Thomas~L. Draper and Feras~A. Saad.
\newblock Efficient {{Online Random Sampling}} via {{Randomness Recycling}},
  July 2025.

\bibitem[DS25b]{DrS25}
Thomas~L. Draper and Feras~A. Saad.
\newblock Efficient {{Rejection Sampling}} in the {{Entropy-Optimal Range}},
  April 2025.

\bibitem[Dud14]{Dud14}
Jarek Duda.
\newblock Asymmetric numeral systems: Entropy coding combining speed of
  {{Huffman}} coding with compression rate of arithmetic coding, January 2014.

\bibitem[Eli72]{Eli72}
Peter Elias.
\newblock The efficient construction of an unbiased random sequence.
\newblock {\em The Annals of Mathematical Statistics}, 43(3):865--870, 1972.

\bibitem[EM99]{ErM99}
P{\'a}l Erd{\"o}s and M~Ram Murty.
\newblock On the order of a (mod p).
\newblock In {\em {{CRM}} Proceedings and Lecture Notes}, volume~19, pages
  87--97, 1999.

\bibitem[HH97]{HaH97}
T.~S. Han and M.~Hoshi.
\newblock Interval algorithm for random number generation.
\newblock {\em IEEE Transactions on Information Theory}, 43(2):599--611, March
  1997.

\bibitem[HV24]{HuV24}
Mark Huber and Danny Vargas.
\newblock Optimal rolling of fair dice using fair coins, December 2024.

\bibitem[Koz14]{Koz14}
Dexter Kozen.
\newblock Optimal {{Coin Flipping}}.
\newblock In Franck Van~Breugel, Elham Kashefi, Catuscia Palamidessi, and Jan
  Rutten, editors, {\em Horizons of the {{Mind}}. {{A Tribute}} to {{Prakash
  Panangaden}}}, volume 8464, pages 407--426. Springer International
  Publishing, Cham, 2014.

\bibitem[KY76]{KnY76}
D.~Knuth and A.~Yao.
\newblock {\em Algorithms and Complexity: New Directions and Recent Results},
  chapter The complexity of nonuniform random number generation.
\newblock Academic Press, 1976.

\bibitem[Lum13]{Lum13}
J{\'e}r{\'e}mie Lumbroso.
\newblock Optimal {{Discrete Uniform Generation}} from {{Coin Flips}}, and
  {{Applications}}, April 2013.

\bibitem[Per92]{Per92}
Yuval Peres.
\newblock Iterating {{Von Neumann}}'s {{Procedure}} for {{Extracting Random
  Bits}}.
\newblock {\em The Annals of Statistics}, 20(1), March 1992.

\bibitem[PL05]{PaL05}
Sung-il Pae and Michael~C. Loui.
\newblock Optimal random number generation from a biased coin.
\newblock In {\em Proceedings of the Sixteenth Annual {{ACM-SIAM}} Symposium on
  Discrete Algorithms}, Soda '05, pages 1079--1088, Vancouver, British Columbia
  and USA, 2005. {Society for Industrial and Applied Mathematics}.

\bibitem[Roc91]{Roc91}
J.R. Roche.
\newblock Efficient {{Generation Of Random Variables From Biased Coins}}.
\newblock In {\em Proceedings. 1991 {{IEEE International Symposium}} on
  {{Information Theory}}}, pages 169--169, Budapest, Hungary, 1991. IEEE.

\bibitem[SFRM20]{SFR20}
Feras~A. Saad, Cameron~E. Freer, Martin~C. Rinard, and Vikash~K. Mansinghka.
\newblock The {{Fast Loaded Dice Roller}}: {{A Near-Optimal Exact Sampler}} for
  {{Discrete Probability Distributions}}, June 2020.

\bibitem[VN{\etalchar{+}}63]{Von63}
John Von~Neumann et~al.
\newblock Various techniques used in connection with random digits.
\newblock {\em John von Neumann, Collected Works}, 5:768--770, 1963.

\end{thebibliography}

\end{document}